\documentclass[conference,10pt]{IEEEtran}
\usepackage[T1]{fontenc}
\usepackage{bm}
\usepackage{amssymb,amsmath,amsfonts,amsthm}
\usepackage{mathrsfs}
\usepackage{cite}
\usepackage{array}
\usepackage{tabularx}
\usepackage[table]{xcolor}
\usepackage{multicol, multirow}
\usepackage{color}
\usepackage{mathtools}
\usepackage{subcaption}
\usepackage{tikz,pgf}
\usepackage{verbatim}
\usetikzlibrary{calc,positioning,mindmap,trees,decorations.pathreplacing}
\usepackage{graphicx}
\usepackage{bbm}
\usepackage[utf8]{inputenc}
\usepackage{algorithm}
\usepackage{algorithmic}
\usepackage{hyperref}

\renewcommand{\IEEEQED}{\IEEEQEDopen}

\newtheorem{propi}{Proposition}

\DeclarePairedDelimiter\abs{\lvert}{\rvert}

\usepackage[colorinlistoftodos,prependcaption,backgroundcolor=black!5!white,bordercolor=red]{todonotes}

\makeatletter
\let\oldabs\abs
\def\abs{\@ifstar{\oldabs}{\oldabs*}}

\IEEEoverridecommandlockouts
\ifCLASSINFOpdf
\else
\fi
\usepackage[cmintegrals]{newtxmath}
\begin{document}
	\title{Utilizing 5G NR SSB Blocks for Passive Detection and Localization of Low-Altitude Drones}

	\author{\IEEEauthorblockN{Palatip Jopanya, Diana P. M. Osorio} \IEEEauthorblockA{ {Department of Electrical Engineering,} {Link\"oping University}, Sweden }  \IEEEauthorblockA{E-mail: \{palatip.jopanya, diana.moya.osorio\}@liu.se} }
	
	\maketitle
	
	\begin{abstract}	
	With the exponential growth of the unmanned aerial vehicle (UAV) industry and a broad range of applications expected to appear in the coming years, the employment of traditional radar systems is becoming increasingly cumbersome for UAV supervision. Motivated by this emerging challenge, this paper investigates the feasibility of employing integrated sensing and communication (ISAC) systems implemented over current and future wireless networks to perform this task. We propose a sensing mechanism based on the synchronization signal block (SSB) in the fifth-generation (5G) standard that performs sensing in a passive bistatic setting. By assuming planar arrays at the sensing nodes and according to the 5G standard, we consider that the SSB signal is sent using sweeping beams that are multiplexed in time, with some of them pointing toward a surveillance region where low-altitude drones can be flying. The Cr\'{a}mer-Rao Bound (CRB) is derived as the theoretical bound for range and velocity estimation. {The probabilities of detection and false alarm are investigated using the peak-to-average factor as a threshold under a power-saving scheme.}
	\end{abstract}

\section{Introduction}
\label{Introduction}
The road toward the standardization of integrated sensing and communication (ISAC) has already started with the 3rd Generation Partnership Project (3GPP) presenting a feasibility study with $32$ use cases and potential requirements in Release 19~\cite{TR}. Among others, it is expected that fifth-generation (5G) radio signals can be used to sense the presence or proximity of unmanned aerial vehicles (UAVs) illegally flying in restricted areas (e.g. airports, and military bases), which comes as a challenging while critical use case. 

In ISAC systems, where the same hardware, waveform, and time-frequency resources are reused for both functionalities, a more efficient use of resources can be achieved, thus avoiding interferences that can appear between communication and radar systems operating simultaneously.

One way to implement ISAC systems is in a passive sensing setting, where communication signals are used as illuminators of opportunity for sensing. Although the abundant downlink data symbols may appear suitable as illuminators of opportunity, they are not ideal candidates for a few reasons. 

First, the transmitting beams are focused at random terminal locations by employing i.e., maximum-ratio transmission (MRT) or zero-forcing (ZF) precoders, leaving gaps of weak signal in the surveillance area. 
Furthermore, the dependency on traffic demands limits their availability.

Pilot signals have shown advantageous passive detection performance. For instance, the work in~\cite{9921271} investigates the positioning reference signal (PRS) for radar sensing. Moreover, the synchronization signal (SSB) of 5G NR is the only periodic block of symbols that is well-defined as a fixed-size block in the time-frequency grid. SSB is used for cell search when a user equipment (UE) first enters the coverage area of the system, as it is periodically swept across the entire coverage area of the base station~\cite{5gnr}. Cell search is carried out continuously by UEs moving between cells, when the UE is in connected or idle mode~\cite{5gnr}. Thus, all these characteristics make this block highly appealing for passive sensing in ISAC networks and have been investigated in~\cite{10083170,10200933}. 

Inspired by the promising insights reported in the above works, this paper further proposes a mechanism based on SSB signals for performing passive sensing of UAVs in a surveillance region, which can be also used, for instance, to detect intruders in non-flying zones. We consider a bistatic implementation employing uniform planar arrays (UPAs) at both transmitter and sensing receiver base stations. Our mechanism considers that a number of beams during the SSB procedure are directed toward the surveillance area at higher altitudes, while the rest are directed toward UEs at lower altitudes. The transmit signals are {SSB signals}, precoded into a sweeping beams that sweep across the entire coverage, in a time-multiplexed fashion. For this system, we derive the Cram\'{e}r-Rao bound (CRB) for range and velocity estimation and compare it with the two-dimensional fast Fourier transform (2D-FFT) estimation method. For the 2D-FFT, we consider a receiver beam aggregation criterion based on a peak-to-average factor (PAF) of the range-velocity profile for the estimation. Additionally, we provide insights on the impact of system parameters and evaluate the impact of the number of beams directed toward the surveillance region over the detection performance.

\section{System model}

\begin{figure*}[htbp]
    \centering
    \hspace{0cm}
    \vspace{0cm} 
    \begin{subfigure}[b]{0.2\textwidth}
        \centering
        \includegraphics[width=2.3\textwidth]{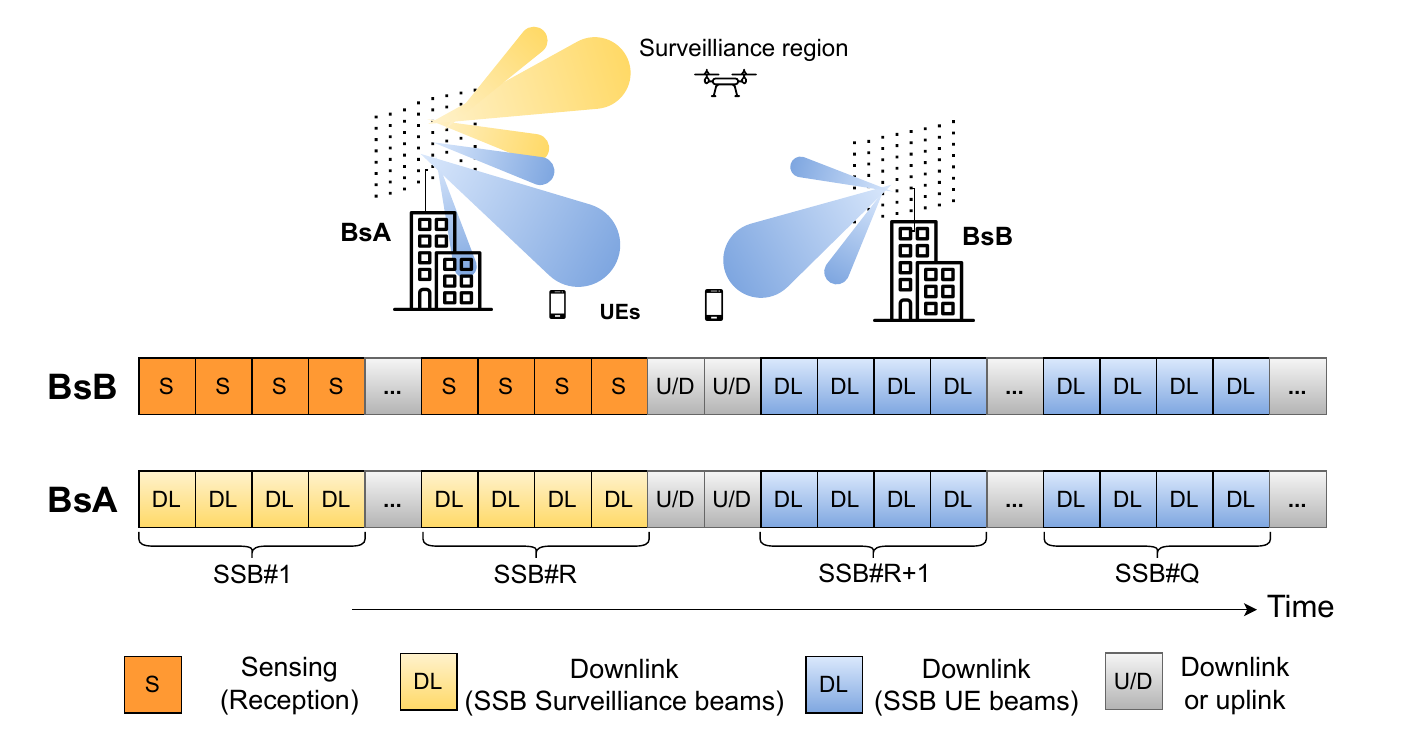}
        \vspace{0.5cm}
        \caption{}
        \label{fig:TDD_workflow}
    \end{subfigure}
    \hspace{4.5cm}
    \begin{subfigure}[b]{0.5\textwidth}
        \centering
        \includegraphics[trim=0cm 0cm 0cm 0cm, clip, width=1\textwidth]{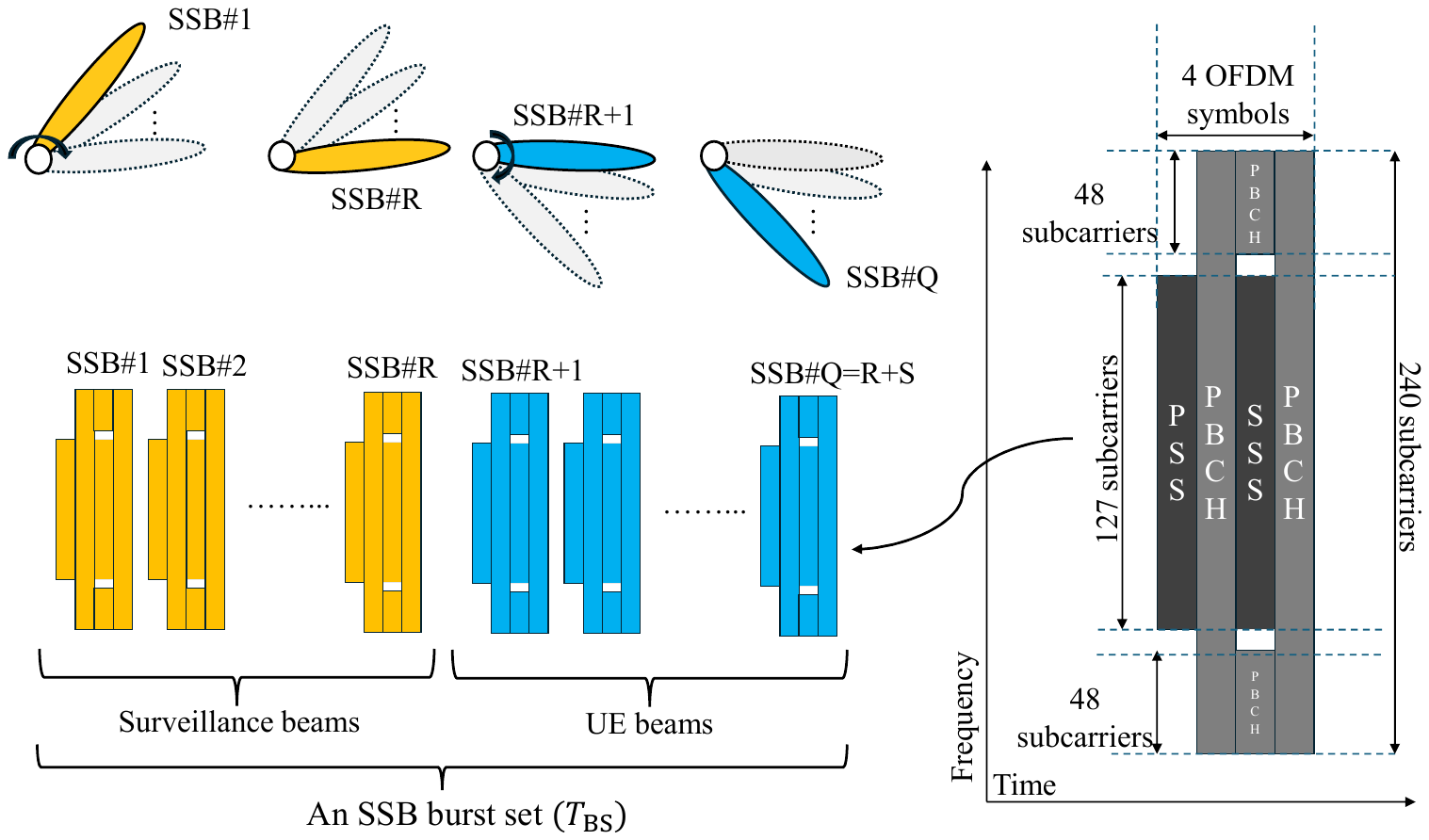}
        \vspace{-0.7cm} 
        \caption{}
        \label{fig:SSB_per}
    \end{subfigure}%
    \caption{a) A bistatic pair and TDD flow with sensing reception, b) An SSB scheduling scheme.}
\end{figure*}
We consider a pair of multiple-input multiple-output (MIMO) ISAC base stations (BSs) that jointly form a bistatic sensing pair consisting of base station A (BsA), which functions as the sensing transmitter, and base station B (BsB), which functions as the sensing receiver, while both BSs are capable of transmitting and receiving communication symbols. We assume that BsA and BsB are coordinated and synchronized in phase and frequency. BsA and BsB are equipped with a UPA with $M_\text{H}$ antennas per row and $M_\text{V}$ antennas per column with $M_\text{V} $$=$$ M_\text{H}$ that add up to $M$$=$$M_\text{V} \times M_\text{H}$ total antennas. The system operates in {a scheme requires an adaptation on the time-division duplexing (TDD) workflow}, where BsB switches to reception mode when BsA transmits sensing symbols as illustrated in Fig. \ref{fig:TDD_workflow}. The BSs serve user equipments (UEs) via data symbols and synchronize with UEs via SSB symbols. We propose that the BsB exploits the echoes from the beam sweeping of the SSB beams directed toward the surveillance area, referred to as surveillance beams. The illustration of the bistatic setup in the top part of Fig. \ref{fig:TDD_workflow} shows the beam sweep of surveillance beams in yellow and the SSBs of UEs in blue. {Note that the pattern can be formed by alternating the receiving base stations when considering only two, thus balancing the overhead.} 



\subsection{Synchronization Signal Block (SSB)}
In the 5G standard, an SSB is a periodic signal that is structured as a block of orthogonal frequency-division multiplexing (OFDM) symbols consisting of PSS (Primary Synchronization Signal), SSS (Secondary Synchronization Signal), and Physical Broadcast Channel (PBCH) symbols, as shown in Fig. \ref{fig:SSB_per}. Each block is transmitted consecutively in time with a few symbols as a gap in different directions and time. A set of {SSB signals}, called an SSB burst set, covers the entire surveillance area and UEs coverage. This beam sweeping process is confined within SSB burst set period ($T_{\text{BS}}$) and will be repeated periodically in every SSB periodicity ($T_\text{SP}$), as shown in Fig. \ref{fig:SSB_per} \cite{5gnr}.
\subsection{Transmit signal}
An OFDM block with $N$ subcarriers and $L$ symbols is considered. The baseband signal of the $n$th subcarrier and $l$th symbol can be written as~\cite{multicarrier} ${x}_{n,l}(t) $$=$$ \sqrt{\rho}x_{n,l} e^{j2\pi n f_\Delta t}\text{rect}\left( \frac{t - lT_{s}}{T_{s}} \right)$, where $n{\in}[1,2,...,N]$ is the subcarrier index, $l {\in} [1,2,...,L]$ is the symbol index, $x_{n,l}$ is the complex data corresponding to the $n$th subcarrier and $l$th symbol, $\rho$ is the transmitted power per symbol, $f_\Delta$ is subcarrier spacing, $T_{s}={1}/{f_\Delta}$ is symbol duration, $T= T_{s}+ T_c$ is the total symbol duration, and $T_c$ is the cyclic prefix duration. We assume that the cyclic prefix is larger than the delay spread of the targets. The antenna steering vector of $r$th beam, $\mathbf{a} (\theta_r,\phi_r) \in \mathbb{C}^{M}$, is a half-wavelength spaced UPA antenna response, where $\theta_r$ is the azimuth and $\phi_r$ is the elevation of $r$th beam. This is defined as~\cite{björnson2024introduction} $\mathbf{a} (\theta_r,\phi_r) = \mathbf{a}_{M_\text{V}}(\phi_r,0)\otimes \mathbf{a}_{M_\text{H}}(\theta_r,\phi_r)$, with $\mathbf{a}_{M_\text{V}} (\phi_r,0) = [1, e^{-j\pi \sin(\phi_r)},..., e^{-j \pi  ({M}_V-1)\sin(\phi_r)}]^T$, 
$\mathbf{a}_{M_\text{H}} (\theta_r,\phi_r) = [1, e^{-j \pi \sin(\theta_r)\cos(\phi_r)},..., e^{-j \pi ({M}_H-1)\sin(\theta_r)\cos(\phi_r)}]^T$, where $\mathbf{a}_{M_\text{H}} (\theta_r,\phi_r) \in \mathbb{C}^{M_H}$ and $\mathbf{a}_{M_\text{V}} (\phi_r,0) \in \mathbb{C}^{M_V}$.
{The SSB signals are represented using sweeping beams, and each beam is associated to a unique angular pair $(\theta_i,\phi_j)$ in the set of antenna responses as}
\begin{equation}
\label{eqn:setGoB}
\{ \mathbf{a} (\theta_i,\phi_j) \} : \: \theta_i =  \arcsin{\left( \frac{2q}{\sqrt{M}} \right)},\phi_j = \arcsin{\left( \frac{2q}{\sqrt{M}} \right)},    
\end{equation}
for $q $$\in$$ [ 0,\pm1,\pm2,.., \lfloor {\sqrt{M}}/{2}$$-$$1  \rfloor ]$, where each antenna response corresponds to each beam. Here we consider $M$$=$$100$, which results in $9$ {angles} in degrees as $\theta_i, \phi_j $$\in$$ \{0,\pm11.5,\pm23.6,\pm36.9,\pm53.1\}$, yielding a total of $Q$$=$$81$ antenna responses in a set. Define the normalized precoder matrix $\mathbf{F} \in \mathbb{C}^{M \times R}$ as
\begin{equation}
\label{eqn:fmat}
\mathbf{F} = \frac{1}{\sqrt{M}} \left[\mathbf{a}^\ast (\theta_1,\phi_1), \mathbf{a}^\ast (\theta_2,\phi_2),..., \mathbf{a}^\ast (\theta_R,\phi_R)\right] ,
\end{equation}
\noindent
where each column represents a unique, conjugate-normalized antenna response in a given set in (\ref{eqn:setGoB}). Note that we consider only a total of $R$$=$$45$ beams in the surveillance area, with an elevation greater than or equal to zero. The passband transmitted signal can be written as $\text{Re} \left\{ {x}_{n,l,r}(t)  \mathbf{f}_r e^{j2\pi f_c t} \right\},$ where $f_c$ is the carrier frequency, $\mathbf{f}_r$ is the $r$th column in matrix $\mathbf{F}$ in (\ref{eqn:fmat}), and ${x}_{n,l,r}(t)$ is a transmit signal of $r$th beam. Define $\mathbf{\Phi} \in \mathbb{C}^{R\times R}$, a diagonal matrix with sequence of unit power symbols in the diagonal, $\mathbf{\Phi}\mathbf{\Phi}^H = \mathbf{I}_R$. This yields $\text{Re} \left\{\mathbf{F} \mathbf{\Phi} e^{j2\pi f_c t} \right\}$, a sequence of transmitting symbols in a sweeping beams. {This paper considers only the SSB signals used for surveillance; the SSBs intended for UEs are not the main focus. The overhead from sensing reception in the TDD frame is approximately $(4$$\cdot$$100$$\cdot$$R$$\cdot$$T)/T_\text{SP}$ percent of the total duration $T_\text{SP}$, and it applies only to the subcarrier of the SSB signal.}
\subsection{Received signal}
We consider only the echoes from surveillance beams, where the clutter from non-target objects in both line-of-sight (LoS) and non-line-of-sight (NLoS) conditions is assumed to be removed by the system. A study in \cite{10840933} demonstrates an OFDM-based clutter rejection method. The channel consists of the LoS direct-link between BsA and BsB and the echoes from the targets. The received signal, $\mathbf{y}_{n,l,r} \in \mathbb{C}^{M}$, of $n$th subcarrier, $l$th symbol, and $r$th beam with a delay shift $\tau$ and a Doppler effect $f_d$ in discrete time is given by \cite{7952790},\cite{multicarrier} 
\begin{equation}
\label{eqn:receivedSig0}
\mathbf{y}_{n,l,r} = x_{n,l,r}(lT_s-\tau) \mathbf{Hf}_r e^{j2\pi f_dlT_s} + \mathbf{w_{n,l}} ,
\end{equation}
where $\mathbf{H} \in \mathbb{C}^{M\times M}$ is the channel, and $\mathbf{w}_{n,l} \sim \mathcal{CN}(0,\sigma_n^2 \mathbf{I})$ is additive white Gaussian noise (AWGN). The transmit symbols are removed at the BsB assuming that
they are known by coordination. Thus, $\mathbf{Y}_{nl}\mathbf{\Phi}^H = \sqrt{\rho}\mathbf{H F}\mathbf{\Phi \Phi}^H + \mathbf{W}' = \sqrt{\rho}\mathbf{H} \mathbf{F} + \mathbf{W}'$, where $\mathbf{Y}_{n,l} \in \mathbb{C}^{M \times R}$ is a matrix of the received sequence, and $\mathbf{W}' = \mathbf{W\Phi}^H$. The received signal, $\mathbf{Y}'_{n,l} = \mathbf{Y}_{nl}\mathbf{\Phi}^H , \; \mathbf{Y}'_{n,l} \in \mathbb{C}^{M \times R}$, can be derived to account for the direct link and the echoes as follows: 
\begin{equation}
\begin{split}
\label{eqn:receivedSig1}
\mathbf{Y}'_{n,l} = \underbrace{\sqrt{\rho\beta_0} \mathbf{H}_0 \mathbf{F} e^{-j2\pi (f_c + nf_\Delta ) \tau_0}}_\text{Direct link}
\\
+ \underbrace{ \sqrt{\rho}  \sum_{k=1}^{K} \alpha_k \sqrt{\beta_k} 
\underbrace{\mathbf{a}
(\theta_{a,k},\phi_{a,k})}_\text{Arrival}
\underbrace{\mathbf{a}^T(\theta_{d,k},\phi_{d,k})}_\text{Departure} \mathbf{F} }_\text{echoes}  
\\
\underbrace{e^{-j2\pi (f_c + nf_\Delta ) \tau_k}}_\text{Delay shift} 
\underbrace{e^{j2\pi f_{d,k} lT_{s}}}_\text{Doppler effect} 
  + \mathbf{W}'_{n,l} ,
\end{split}
\end{equation}
\noindent
where $K$ is the number of targets, $\beta_0$ is the channel gain from direct link, $\mathbf{H}_0$ is the direct link channel, $\theta_{a,k},\phi_{a,k}$ are arrival azimuth and arrival elevation of $k$th target, respectively, corresponding to BsB, $\theta_{d,k},\phi_{d,k}$ are departure azimuth and departure elevation of $k$th target, respectively, corresponding to BsA, $\tau_k$ is the total delay of the echoes of target $k$th, $f_{d,k}$ is the Doppler shift effect from target $k$th, and $\alpha_k \sim \mathcal{CN}(0,1)$ is the randomness model of a radar cross section (RCS) of the $k$th target that follows a Swerling 1 model \cite{li2025detectingunauthorizeddronescellfree}, where the realization fluctuates slowly and remains fixed within the OFDM block length $L$. $\beta_k$ is the large-scale fading of the echoes, BsA-target $k$th-BsB, that is
    \begin{equation}
    \label{eqn:betak}
    \beta_k = \frac{\lambda_c^2 \sigma_{\text{rcs},k}}{(4\pi)^3d^2_{tx,k}d^2_{k,rx}},
    \end{equation}
 \noindent
where $\lambda_c$ is a wavelength of the carrier frequency, $d_{tx,k}$ is the distance between BsA and target $k$th, $d_{k,rx}$ is the distance between target $k$th and BsB, and $\sigma_{rcs}$ is the RCS in dBsm. {Note that the sensing is done without receive beamforming.} We assume that the first term (direct link) in (\ref{eqn:receivedSig1}) is known and can be removed at BsB since both BsA and BsB are stationary. The received signal, $\textbf{Y}'_{n,l}$, after direct link cancellation, down-conversion and substituting the Doppler shift with velocity ($f_{d,k} ={2v_k}/{\lambda_c}$) and the delay with the bistatic range ($\tau = {d_k}/{c}$) is
\begin{equation}
\begin{split}
\label{eqn:receivedSig2}
\mathbf{Y}'_{n,l} = 
\sqrt{\rho}  \sum_{k=1}^{K} \alpha_k \sqrt{\beta_k} 
\mathbf{a}
(\theta_{a,k},\phi_{a,k})
\mathbf{a}^T(\theta_{d,k},\phi_{d,k}) \mathbf{F} 
\\
e^{-j2\pi nf_\Delta \frac{d_k}{c}} 
e^{j4\pi \frac{v_k}{\lambda_c} lT_{s}}
 + \mathbf{W}'_{n,l},
\end{split}
\end{equation}

\noindent
where $c$ is the speed of light, $d_k$ is the bistatic range of target $k$th, $v_k$ is the relative radial velocity of target $k$th with respect to the BsB. The unambiguous range is a detectable distance given by the propagated distance of one symbol duration as $d_{\text{u}} \leq {c T_{s}}$. The unambiguous velocity is the range of maximum and minimum relative radial velocities, which can be defined as $\vert v_{\text{u}} \vert \leq {\lambda_c f_\Delta}/{2 }$.

\section{Parameters estimation}
\subsection{2D-FFT}
The two-dimensional fast Fourier transform (2D-FFT) is a non-parametric method in spectral estimation \cite{Braun2014OFDMRA}. We use 2D-FFT for estimating the range and velocity and to compare it with the CRB. This method is used because 2D-FFT is simple and does not require the number of targets a priori. To implement 2D-FFT in this system, we first form a received frame matrix, $\mathbf{Z}_r \in \mathbb{C}^{N \times L}$, of SSB beam $r$th disregarding the dimension of $M$ (multiple antennas). Given that the received signal vector $\mathbf{y}'_{n,l,r} = [y_{n,l,r,1},y_{n,l,,r,2},...,y_{n,l,r,M}]^T$ is the $r$th column of $\mathbf{Y}'_{n,l}$ in (\ref{eqn:receivedSig2}). $\mathbf{Z}_r$ is formed by arranging the first element of each $\mathbf{y}'_{n,l,r}$ as
\begin{equation}
\label{eqn:Zmat}
\mathbf{Z}_r =
\begin{bmatrix}
 y_{1,1,r,1} & y_{1,2,r,1} & ... & y_{1,L,r,1} \\
 y_{2,1,r,1} & y_{2,2,r,1} & ... & y_{2,L,r,1} \\
 \vdots & \vdots & \ddots & \vdots \\
 y_{N,1,r,1} & y_{N,2,r,1} & ... & y_{N,L,r,1} \\
\end{bmatrix}.
\end{equation}
\noindent
Define a range-velocity profile as $\mathbf{Z}'_r \in \mathbb{R}^{N' \times L'}$, the $n'$th row and $l'$th column in matrix $\mathbf{Z}'_r$ is 
\begin{equation}
    \label{eqn:perio}
    \left[ \mathbf{Z}'_{r} \right] _{n',l'} = \abs{ \sum_{n=0}^{N'-1} \left( \sum_{l=0}^{L'-1}\left[ \mathbf{Z}_{r} \right] _{n,l}e^{-j2\pi l'\frac{l}{L'}} \right) e^{j2\pi n' \frac{n}{N'}} },
\end{equation}
where $\left[ \mathbf{Z}_{r} \right]_{n,l}$ is the $n$th row and $l$th column in matrix $\mathbf{Z}_r$. The range-velocity profile in (\ref{eqn:perio}) can be seen as row-wise $L'$-point FFT and column-wise $N'$-point IFFT of $\mathbf{Z}_r$ with padding zeros, $L'$$>$$L, \; N'$$>$$N$.
\subsection{Received beam aggregation}
{We consider the received beam aggregation with a non-coherent integration of beams approach, where the power of multiple received beams are added up to improve the performance}. {As BsB receives a total of $R$ OFDM blocks where only a few contain target echoes, we propose an aggregation method based on a threshold that we call PAF, which is defined as} $p_r $$=$$ 
{\max_{n',l'} \mathbf{Z}'_{r} }/({\mathbf{1}^T\mathbf{Z}'_{r} \mathbf{1}})$ and the PAF vector is $\text{PAF} $$=$$ [p_1, p_2,...p_R]^T$. Define a set $\eta$ that satisfies $\eta = \{i\}, \;i\  \text{s.t.} \  p_i > \gamma_a, \; i \in \{1,2,...,R \}$, where $\gamma_a$ is a selected threshold. A larger $\gamma_a$ enhances the separation between the target and the noise. Define $\mathbf{Z'}$, the aggregated profile, as $\mathbf{Z'} = \frac{1}{\abs{\eta}} \sum_{j \in \eta} \mathbf{Z}'_j$, where $\abs{\eta}$ is the size of the set $\eta$. For a single target, the bistatic range estimate ($\hat{d}$) and relative velocity estimate ($\hat{v}$) can be found as
    $\hat{d}={\hat{n}c T_s}/{N'} ,\ \hat{v}={\hat{l}f_\Delta \lambda_c}/({2 L'}),$
where $\hat{n}$ and $\hat{l}$ are the row and column indices that locate the peak in $\mathbf{Z}'$ as
$(\hat{n},\hat{l}) = \operatorname*{argmax}_{{n'},{l'}} \mathbf{Z}'$.
\subsection{Target detection}
The detection decision is based on the threshold, $\gamma$.  The detector is given as $p \underset{\mathcal{H}_0}{\overset{\mathcal{H}_1}{\gtrless}}\gamma$, where, the detector indicates the presence of a target when $p>\gamma$, $\mathcal{H}_1$, and the absence of a target if $p<\gamma$, $\mathcal{H}_0$. Here, $p$ is defined as $p= \max_r \text{PAF}$.
\section{Cram\'{e}r-Rao Bound}
For simplicity, we consider the CRB of the random variables, range and velocity, denoted as $\Theta $$=$$ [\Theta_i,\Theta_j]^T$$=$$ [d,v]^T$, in a single target case ($K$$=$$1$). We evaluate this for a given signal-to-noise ratio at BsB, $\text{SNR}_r={\rho \beta g^2}/({\sigma_n^2 M})$, assuming a given RCS and fixed flight path. The received signal at the $n$th row, $l$th column of (\ref{eqn:Zmat}), disregarding the subscript $r,1$ and $k$, can be rewritten as  
\begin{equation}
\label{eqn:receivedSig4}
{y}_{n,l} = \sqrt{\frac{\rho}{M} \beta} \alpha g   e^{-j2\pi nf_\Delta \frac{d}{c}} e^{j4 \pi \frac{v}{\lambda_c} lT_{s}} + {w}_{n,l}, \: {y}_{n,l} \in \mathbb{C},
\end{equation}
\noindent
where $g$$=$$\mathbf{a}^T(\theta_d,\phi_d) \mathbf{a}^*(\theta_{r},\phi_{r})$ is the beamforming gain of the $r$th transmitting beam with respect to the target's angle. We consider the case with $M$=$100$ and $Q$=$81$ in (\ref{eqn:setGoB}) and (\ref{eqn:fmat}). The gain $\abs{g}$ is then bounded by $M(\text{dB}) - 3.7 \text{dB} < \abs{g} \leq M(\text{dB}) ,
\label{eqn:g_bound}$ where $-3.7$ dB is from the selected $M$, and the bound holds irrespective of whether the target is located inside the sweeping beams, given that the $r$th beam is directed closest to the target among all $R$ beams. Formulate a vector $\mathbf{y}$$=$$[y_{1,1},y_{2,1},\dots,y_{2,1},y_{2,2},\dots,y_{N,L}]^T, \; \mathbf{y} \in \mathbb{C}^{NL}$, from (\ref{eqn:receivedSig4}), i.e., by flattening $\mathbf{Y}$, where $[\mathbf{Y}]_{n,l}$$=$$y_{n,l}$. The Fisher information, $\text{FIM}(\Theta)$, for a complex Gaussian signal, as given in \cite[Eq. (15.52)]{10.5555/151045}, is defined as 
\begin{equation}
\begin{split}
\label{eqn:receivedSig5} 
{ \left[ \text{FIM}(\Theta) \right] }_{i,j} =
2\text{Re} \biggl\{ \frac{\partial \mathbf{\mu_y}^H (\Theta)}{\partial \Theta_i} \mathbf{C}_y^{-1}(\Theta) \frac{\partial \mathbf{\mu_y} (\Theta)}{\partial \Theta_j}  \biggl\},
\end{split}
\end{equation}
\noindent
where $\mathbf{C_y} $$= $$\sigma_n^2\mathbf{I}$ is the covariance matrix of $\mathbf{y}$. $\mathbf{\mu_y} $$ =$$ \sqrt{\frac{\rho}{M} \beta} \alpha g [e^{-j2\pi f_\Delta \frac{d}{c}} e^{j4 \pi \frac{v}{\lambda_c} T_{s}},\dots,e^{-j2\pi Nf_\Delta \frac{d}{c}} e^{j4 \pi \frac{v}{\lambda_c} LT_{s}}]^T$ is the mean of $\mathbf{y}$. Note that $\alpha$ is constant within an OFDM block length $L$ (Swerling 1 model).
\begin{propi}
The lower bound on the variance of the distance estimates is
\begin{equation}
\begin{split}
\label{eqn:vard} 
var\{\hat{d}\} \geq \frac{3 T_s^2 c^2 (2L+1) }{ \pi^2 \text{SNR}_r NL ( 7NL-N-L-5 ) (N+1)},
 \end{split}
\end{equation}
\noindent
and the bound on the variance of the velocity estimates is
\begin{equation}
\begin{split}
\label{eqn:varv} 
var\{\hat{v}\} \geq \frac{3 \lambda_c^2 f_\Delta^2 (2N+1)}{4 \pi^2 \text{SNR}_r NL ( 7NL-N-L-5 )(L+1)},
 \end{split}
\end{equation}
\noindent
where $\text{SNR}_r$ is the SNR at BsB after attenuation from pathloss and $\sigma_\text{RCS}$.
\end{propi}
\begin{proof}
See appendix A.
\end{proof}


\section{Simulation Results}
\begin{table}
    \caption{Simulation parameters.}
    \begin{center}
        \begin{tabular}{c|c|c|c   }
            \hline
            \textbf{Parameter}  & \textbf{Value} & \textbf{Parameter}  & \textbf{Value} \\
            \hline
            $M_\text{V} $$=$$ M_\text{H}$ & $10$ &      $\sigma_\text{RCS}$ & $-10$ dBsm \\
            \hline
            $f_\Delta$ & $60$ kHz & $\gamma_a$ & $6$ \\
            \hline
            $f_c$ & $15$ GHz & R & 45
            \\ 
            \hline
            
        \end{tabular}
        \label{tab1}
    \end{center}
    \label{table:simparams}
\end{table}
The parameters and values are provided in Table \ref{table:simparams}. Fig. \ref{fig:CRB_theo} shows the CRB for different configurations of $N$ and $L$. A larger $N$ lowers the CRB bound for both range and velocity, with a more significant impact on range. Similarly, a larger $L$ lowers the CRB for both range and velocity, with a greater impact on velocity.
\begin{figure}[h]
    \centering
    \subfloat[Range CRB.]{%
        \includegraphics[width=1.68in]{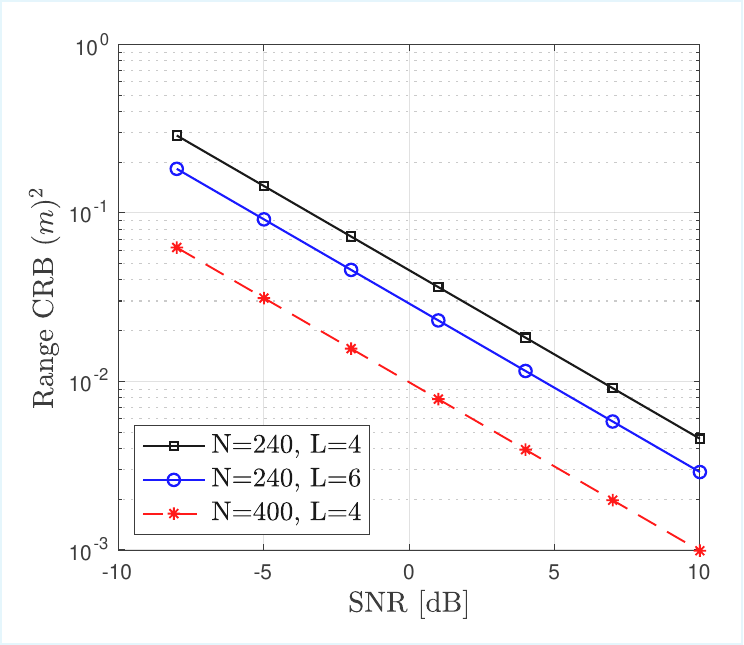}
    }
    \hfill
    \subfloat[Velocity CRB.]{%
        \includegraphics[width=1.68in]{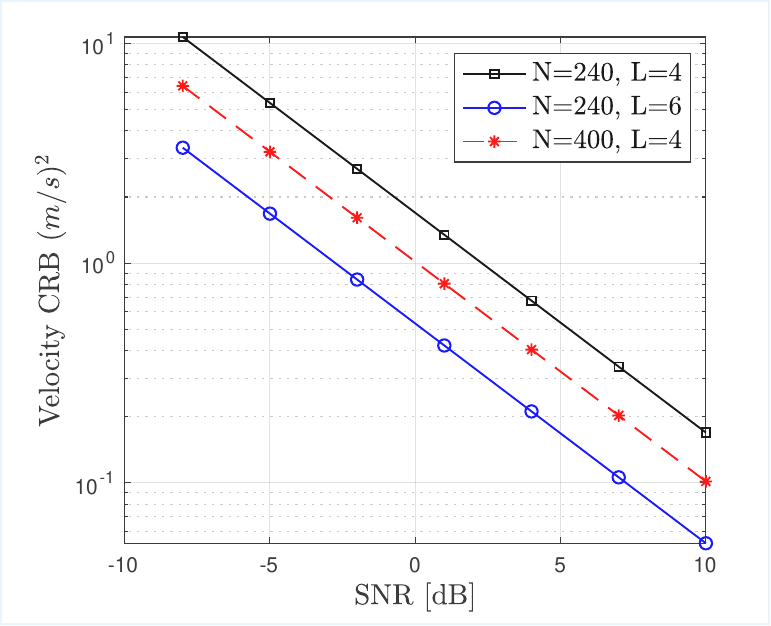}
    }
    \caption{Range and velocity CRB in different block sizes.}
    \label{fig:CRB_theo}
\end{figure}
In Fig. \ref{fig:range_veloc_sim_crb}, a target is generated with a random range, $d \leq d_u$, and random velocity, $\abs{v}\leq v_u$. Fig. \ref{fig:range_veloc_sim_crb}a shows CRB and range root mean square error (RMSE) of an $N'$-point IFFT at SNR$_r$$=$$ -10$ dB, where the SSB line represents the received signal nullified at some symbols in the block, as shown in Fig. \ref{fig:SSB_per}. Similarly, Fig. \ref{fig:range_veloc_sim_crb}b shows the CRB and velocity RMSE of $L'$-point FFT at SNR$_r$$=$$ -7$ dB. The results show that RMSE values for $N$$=$$240$, $ L$$=$$4$, and the {SSB signals} are approximately the same for both cases, range and velocity. The absence of some symbols in the {SSB signals} has minimal effect compared to the completed block.
\begin{figure}[h]
    \centering
    \subfloat[$N'$-point IFFT range RMSE \\at SNR$_r= -10$ dB.]{%
        \includegraphics[width=1.68in]{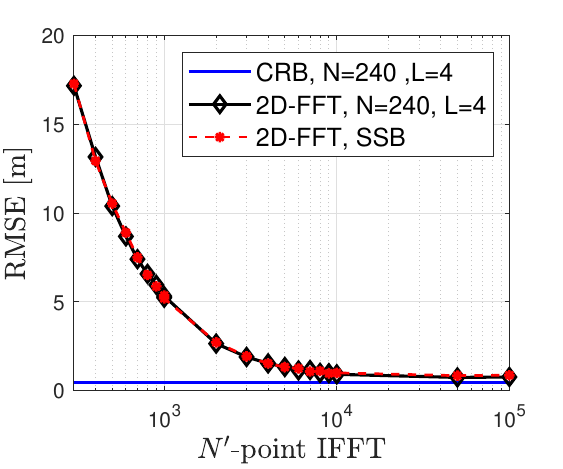}
    }
    \hfill
    \subfloat[$L'$-point FFT velocity RMSE at SNR$_r= -7$ dB.]{%
        \includegraphics[width=1.68in]{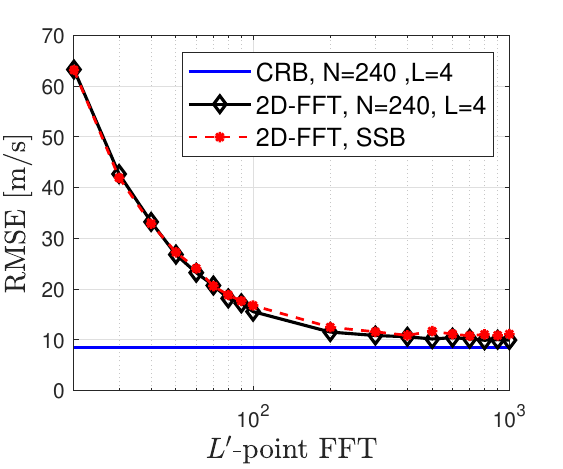}
    }
    \caption{Range and velocity RMSE.}
    \label{fig:range_veloc_sim_crb}
\end{figure} 

{Furthermore, the detection performance is investigated by uniformly and randomly deactivating different percentages of surveillance beams. To that, let us define $P_{fa}$$=$$ P(\max_r \text{PAF}>\gamma  |  \mathcal{H}_0)$ as the probability of false alarm and $P_d$$=$$ P(\max_r \text{PAF}>\gamma  |  \mathcal{H}_1)$ as the probability of detection. Thus, a $100$\% beam deactivation indicates that all surveillance beams are turned off. Fig. \ref{fig:music_deac} shows that with an appropriate thresholds i.e., $\gamma $$=$$ 4$, $P_d$ is maintained above $0.7$ up to $20$\% beam deactivation with $P_{fa}$$=$$0.03$.}
\begin{figure}[h]
    \centering
    {%
        \includegraphics[width=2.3in]{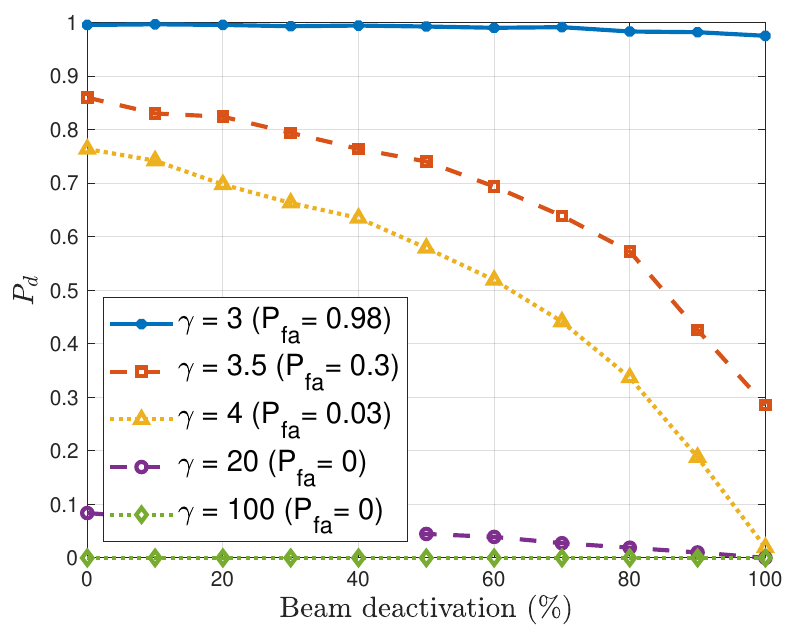}
    }
    \caption{$P_d$ during beam deactivation (\%).}
    \label{fig:music_deac}
\end{figure}
\section{Conclusions}
{This paper investigates 
the feasibility of sensing low-altitude drones by employing the sweeping beams of SSB signals of the mobile network. A set of SSB signals is used as illuminators in a bistatic sensing setup (surveillance beams), and a base station serves as the sensing node. For detection, the decision relies on the threshold and is evaluated with the highest peak-to-average factor (PAF) of the 2D-FFT across the sweeping beams, while localization relies on the aggregated 2D-FFT, assisted by the PAF. Simulation results for the proposed scenario demonstrate the potential for drone-like target detection and localization.}
\label{Conclusions}

\appendices
\section{Proof of proposition 1}
The Fisher information for each element is
\begin{align}
\begin{split}
\label{eqn:fim_dd} 
{ \left[ \text{FIM}(\Theta) \right] }_{d,d} &=2\text{Re} \biggl\{ 4 \pi^2 \frac{\rho \beta g^2}{M \sigma_n^{2}} \mathbb{E}\{\abs{\alpha}^2\}  \frac{f_\Delta^2}{c^2} L \sum_{n=1}^N n^2  \biggl\} \\
&= \frac{4}{3} \pi^2 \text{SNR}_r \frac{f_\Delta^2}{c^2}  NL (N+1)(2N+1),
\end{split}
\end{align}
\begin{equation}
\begin{split}
\label{eqn:fim_vv} 
{ \left[ \text{FIM}(\Theta) \right] }_{v,v} =2\text{Re} \biggl\{ 16 \pi^2  \frac{\rho \beta g^2}{M \sigma_n^{2}} \mathbb{E}\{\abs{\alpha}^2\}  \frac{T_s^2}{\lambda_c^2} N \sum_{l=1}^L l^2  \biggl\}
 \\
= \frac{16}{3} \pi^2 \text{SNR}_r \frac{T_s^2}{\lambda_c^2} N L(L+1)(2L+1),
\end{split}
\end{equation}
\begin{equation}
\begin{split}
\label{eqn:fim_dv} 
{ \left[ \text{FIM}(\Theta) \right] }_{d,v} =2\text{Re} \biggl\{ -8 \pi^2 \frac{\rho \beta g^2}{M \sigma_n^{2}} \mathbb{E}\{\abs{\alpha}^2\}  \frac{f_\Delta T_s}{c \lambda_c} \sum_{n=1}^N \sum_{l=1}^L nl  \biggl\}
 \\
= -4 \pi^2 \text{SNR}_r \frac{1}{c \lambda_c} NL(L+1)(N+1),
\end{split}
\end{equation}
\noindent
where $[\text{FIM}(\Theta) ]_{v,d}$$=$$[\text{FIM}(\Theta) ]_{d,v}$, $\mathbb{E}\{\abs{\alpha}^2\} $$=$$1$, $\text{SNR}_r$$=$${\rho \beta g^2}/({\sigma_n^2 M})$, $\sum_{n=1}^N n $$=$$ {N(N+1)}/{2}$ ,$\sum_{n=1}^N n^2 $$=$$ {N(N+1)(2N+1)}/{6}$ and $T_s f_\Delta $$=$$1$.
\noindent
The Fishier information matrix is
\begin{equation}
\begin{split}
\label{eqn:receivedSig6} 
{  \text{FIM}(\Theta) } = 
4 \pi^2 \text{SNR}_r NL
\\
\begin{bmatrix}
 \frac{1}{3}\frac{f_\Delta^2}{c^2}(N+1)(2N+1) &  -\frac{1}{c \lambda_c} (L+1)(N+1) \\
 - \frac{1}{c \lambda_c} (L+1)(N+1) & \frac{4}{3}\frac{T_s^2}{ \lambda_c^2} (L+1)(2L+1)
\end{bmatrix}.
\end{split}
\end{equation}
Finally, the $\text{CRB}(\Theta )$ is derived as
\begin{equation}
\begin{split}
\label{eqn:crb1} 
 {\text{CRB}(\Theta )}  =   {\text{FIM}^{-1}(\Theta)} = \frac{ 3}{ \pi^2 \text{SNR}_r NL ( 7NL-N-L-5 )}
 \\
\begin{bmatrix}
T_s^2 c^2 \frac{(2L+1)}{(N+1)} &  \frac{3}{4}\lambda_c c \\
 \frac{3}{4}\lambda_c c & \frac{1}{4} \lambda_c^2 f_\Delta^2 \frac{(2N+1)}{(L+1)}
\end{bmatrix}.
 \end{split}
\end{equation}
\noindent
\section*{Acknowledgment}
This work was supported in part by ELLIIT and in part by WASP-funded project ``ALERT''.





\bibliographystyle{IEEEtran}
\bibliography{ref}
	
\end{document}